\documentclass[a4paper,reqno,10pt]{amsart}

\usepackage[left=3cm,right=3cm,top=4cm,bottom=3cm]{geometry}
\usepackage[T1]{fontenc}
\usepackage[utf8]{inputenc}
\usepackage{lmodern}
\usepackage{amsmath,amssymb,amsthm,mathtools}
\usepackage{microtype}
\usepackage{xcolor}
\usepackage{hyperref}

\numberwithin{equation}{section}

\hypersetup{
 colorlinks=true,
 linkcolor=blue,
 citecolor=blue,
 urlcolor=blue
}

\newtheorem{theorem}{Theorem}[section]
\newtheorem{proposition}[theorem]{Proposition}
\newtheorem{lemma}[theorem]{Lemma}
\newtheorem{corollary}[theorem]{Corollary}
\theoremstyle{definition}
\newtheorem{definition}[theorem]{Definition}
\newtheorem{remark}[theorem]{Remark}

\newcommand{\R}{\mathbb{R}}
\newcommand{\C}{\mathbb{C}}
\newcommand{\Om}{\Omega}
\newcommand{\Dom}{\operatorname{dom}}
\newcommand{\Ran}{\operatorname{ran}}
\newcommand{\Ker}{\operatorname{ker}}
\newcommand{\reg}{\mathrm{reg}}
\newcommand{\dd}{\mathrm{d}}
\newcommand{\pa}{\partial}
\newcommand{\la}{\lambda}
\newcommand{\tr}{\operatorname{tr}}
\newcommand{\dist}{\operatorname{dist}}
\newcommand{\diag}{\operatorname{diag}}
\newcommand{\bftau}{\boldsymbol\tau}

\begin{document}

\title[Robin Laplacians with point interactions]
{Robin Laplacians with point interactions on unbounded domains:
discrete spectrum and coupling asymptotics}

\author[D.~Noja]{Diego Noja}
\address[D.~Noja]{Department of Mathematics and Applications,
University of Milano--Bicocca, Via Roberto Cozzi 55, 20126 Milano, Italy}
\email{diego.noja@unimib.it}

\author[F.~Raso Stoia]{Francesco Raso Stoia}
\address[F.~Raso Stoia]{Independent researcher}
\email{f.rasostoiawork@gmail.com}

\date{\today}

\subjclass[2020]{Primary 35J10, 81Q10; Secondary 47B25, 35P05}

\keywords{Point interactions; Robin Laplacian; boundary triples; Weyl matrix;
critical couplings; coupling asymptotics; unbounded domains}

\begin{abstract}
We investigate the discrete spectrum of finitely many point interactions for
Neumann and Robin Laplacians on special unbounded and exterior $C^{1,1}$
domains in dimensions two and three. The operators are realized as
self-adjoint extensions through an ordinary boundary triple whose gamma field
and Weyl matrix are constructed from the Robin Green kernel. Eigenvalues below
the background spectrum are characterized by the Weyl matrix, yielding an
exact finite-dimensional counting formula. If the background operator is
non-negative, this also gives the number of negative eigenvalues without
assuming a finite zero-energy limit of the Weyl matrix. In the one-centre case
we identify the critical coupling and prove that the unique eigenvalue branch
is real analytic, strictly increasing, and strictly concave. For
scalar multicentre couplings $\Theta=\alpha I_N$, sufficiently
strong attraction produces exactly $N$ eigenvalues below the background
spectrum, and all of them have universal leading asymptotics coinciding with
the whole-space laws. If the bottom of the background spectrum is an isolated
eigenvalue, the branch exists for every finite coupling and we determine its
leading decoupling asymptotics as the coupling tends to $+\infty$.
Explicit exterior-sphere and exterior-disk models illustrate the critical
couplings and their threshold behaviour.
\end{abstract}

\maketitle


\section{Introduction}
Point interactions are a well-known class of singular perturbations of the
Laplacian. Unlike the elementary one-dimensional case, in dimensions two and
three the Dirac measure is not an ordinary form-bounded potential on $H^1$, and
the model must instead be defined through self-adjoint extension theory or
suitably renormalized quadratic forms; see
\cite{AGHH,Posilicano2008,BrueningGeylerPankrashkin2008,GM-SelfAdj_book-2023,DerezinskiPoint}.
The Green function of the domain plays a central role. On a proper subdomain of
$\mathbb R^d$ its leading singularity remains universal, whereas its regular part
contains the dependence on the geometry and on the boundary condition (compare
\cite[Theorems~15--16]{BrueningGeylerPankrashkin2005}). This makes several
spectral-theoretic problems both natural and substantially less explicit than
in the whole space. The literature on point interactions in domains with
boundary is comparatively limited. Point interactions in bounded domains were
studied, among others, in
\cite{BlanchardFigariMantile,ExnerMantile,LotoMiche21}; for point interactions
in quantum waveguides, see \cite{EK2015}.
Resonances and spectral theory for the half-space were treated in
\cite[Propositions~3.1--3.5]{NojaRasoStoia2025}, while the Dirichlet problem on more general unbounded
domains was recently analysed in \cite{NojaRasoStoia2026}, see also references therein.
The purpose of this paper is to continue the analysis in the latter paper by giving a unified
extension-theoretic construction and a detailed description of the spectrum
below the background threshold. We work in dimensions $d=2,3$ on exterior $C^{1,1}$ domains (compact boundary) or unbounded
special $C^{1,1}$ domains (unbounded boundary). The Robin
coefficient is an arbitrary real function
$\beta\in L^\infty(\partial\Omega)$, and the interaction set
$Y=\{y_1,\ldots,y_N\}$ is finite. As in the whole-space model, a
self-adjoint matrix $\Theta$ parametrizes the point-interaction operators
$H^\beta_{\Theta,Y}$, while $H_\beta$ denotes the background Robin
Laplacian.

The Green functions $G^\beta_{z,y_j}$ of $H_\beta-z$ and their regular
parts are constructed in Theorem~\ref{thm:robin-kernel}; the associated
gamma field $\gamma_\beta(z)$ and Weyl matrix function $M^\beta_Y(z)$ are
introduced in Section~\ref{sec:preliminaries}. They determine an ordinary
boundary triple for the restriction of $H_\beta$ to functions vanishing on
$Y$. In Section~\ref{sec:operator} this gives a concise realization of
$H^\beta_{\Theta,Y}$ by the local condition
$\psi_{\reg,Y}=\Theta q$. Corollary~\ref{cor:krein-form} then yields the
Kre\u{\i}n resolvent formula, the finite-rank resolvent difference, equality
of the essential spectra, and the eigenvalue identity
\begin{equation*}
 \Ker(H^\beta_{\Theta,Y}-E)
 =
 \gamma_\beta(E)\Ker\bigl(\Theta-M^\beta_Y(E)\bigr),
 \qquad E\in\rho(H_\beta)\cap\mathbb R.
\end{equation*}
Thus every eigenvalue below the background spectrum is detected by a
finite-dimensional Hermitian matrix.

The strict energy monotonicity of the Weyl matrix is given by
\eqref{eq:weyl-derivative}. It yields the exact counting formula
\begin{equation*}
 \dim\Ran\mathbf 1_{(-\infty,E_0)}(H^\beta_{\Theta,Y})
 =n_-\bigl(\Theta-M^\beta_Y(E_0)\bigr),
 \qquad E_0<\inf\sigma(H_\beta),
\end{equation*}
and, in particular, the bound of $N$ on the number of eigenvalues below the
background spectrum; see Corollary~\ref{cor:multicenter-count}. If
$H_\beta\ge0$, the same corollary gives, without an additional threshold
hypothesis,
\begin{equation*}
 \dim\Ran\mathbf 1_{(-\infty,0)}(H^\beta_{\Theta,Y})
 =
 \lim_{\kappa\downarrow0}
 n_-\bigl(\Theta-M_Y^\beta(-\kappa^2)\bigr).
\end{equation*}
If $M_Y^\beta(-\kappa^2)$ converges in $\mathbb C^{N\times N}$ as
$\kappa\downarrow0$, its finite Hermitian limit is denoted by
$M_Y^\beta(0)$, and the right-hand side is
$n_-(\Theta-M_Y^\beta(0))$.
In particular, if $\Theta=\alpha I_N$ and
$\mu_1\leq\cdots\leq\mu_N$ are the eigenvalues of the finite matrix
$M_Y^\beta(0)$, Corollary~\ref{cor:multicenter-count} gives
\begin{equation*}
 \dim\Ran\mathbf 1_{(-\infty,0)}
 \bigl(H^\beta_{\alpha I_N,Y}\bigr)
 =
 \#\{j:\mu_j>\alpha\}.
\end{equation*}
Hence, unlike the one-centre problem, where a single critical coupling
separates binding from non-binding, in the multicentre problem the number
of negative eigenvalues may change successively as $\alpha$ crosses the
distinct eigenvalues of $M_Y^\beta(0)$. These values are thresholds for
the negative-eigenvalue count; the spectral behaviour at zero itself
requires a separate threshold analysis, here not tackled.

For one centre the spectral picture below the background spectrum is complete.
Proposition~\ref{prop:single-center-spectrum} identifies the critical coupling
and proves that the eigenvalue, whenever it exists, is unique and simple.
Writing it as $E(\alpha)$, Corollary~\ref{cor:one-center-concavity} gives
\begin{equation*}
 \frac{\dd E}{\dd\alpha}
 =\frac{1}{\|G^\beta_{E(\alpha),y}\|_{L^2(\Omega)}^2}>0,
 \qquad
 \frac{\dd^2E}{\dd\alpha^2}<0,
\end{equation*}
with an explicit expression for the second derivative. Hence $E(\alpha)$ is real analytic, strictly increasing,
and strictly concave. For the scalar multicentre coupling
$\Theta=\alpha I_N$, Corollary~\ref{cor:strong-coupling} shows that sufficiently
strong attraction produces exactly $N$ eigenvalues below the background
spectrum. In both dimensions every one of these eigenvalues has the same
universal leading asymptotics as the whole-space one-centre model.
If the bottom of the background spectrum is an isolated
eigenvalue, Corollary~\ref{cor:isolated-background-ground-state} instead
describes the opposite, decoupling regime $\alpha\to+\infty$: the critical
coupling is then $+\infty$, the point-interaction eigenvalue exists for every
finite $\alpha$, and its distance from the background ground-state energy is
determined to leading order by the value of the background eigenfunction at
the interaction centre.

The paper is organized as follows. Section~\ref{sec:preliminaries} constructs
the Robin Green kernel, the gamma field, and the Weyl matrix function.
Section~\ref{sec:operator} gives the boundary-triple realization and the
Kre\u{\i}n formula. Section~\ref{sec:spectrum} contains the subthreshold
counting theorem, the one-centre analysis, and the coupling asymptotics.
Finally, Section~\ref{sec:models} gives explicit critical couplings for an
exterior sphere and an exterior disk. Except in these radial examples, a
classification of threshold states and an analysis of positive-energy
resonances are not attempted.

\section{Preliminaries}\label{sec:preliminaries}

\subsection{Geometric setting}
Throughout this work $d\in\{2,3\}$. We consider one of the following classes
of domains $\Omega\subset\R^d$:

\begin{itemize}
\item[(a)] a \emph{special $C^{1,1}$ domain},
\begin{equation*}
\begin{aligned}
\Omega&=\{(x',x_d)\in \R^{d-1}\times \R:\ x_d>\varphi(x')\},\\
\varphi&\in C^{1,1}(\R^{d-1}),\qquad \|\nabla\varphi\|_{L^\infty}<\infty;
\end{aligned}
\end{equation*}

\item[(b)] an \emph{exterior $C^{1,1}$ domain}, namely a connected open set
with compact complement,
\begin{equation*}
\Omega=\R^d\setminus K,
\end{equation*}
where $K\subset\R^d$ is compact and
$\partial\Omega=\partial K$ is of class $C^{1,1}$.
\end{itemize}

Both classes are examples of unbounded, quasi-conical domains, with respectively unbounded or bounded boundary.

\subsection{The classical Robin form}\label{sec:background}

Throughout the paper the scalar product in $L^2(\Om)$ is linear in the
first variable, and we use the same convention in $\C^N$. The duality
between $H^{-1/2}(\partial\Omega)$ and $H^{1/2}(\partial\Omega)$ is denoted
by $\langle\cdot,\cdot\rangle_{\partial\Omega}$.

Let $\Omega$ be a special $C^{1,1}$ domain or an exterior $C^{1,1}$ domain and $\beta\in L^\infty(\pa\Om)$ be real-valued. The Robin Laplacian $H_\beta$ is the self-adjoint operator in $L^2(\Om)$ associated with the closed semibounded form
\begin{equation*}
 a_\beta[u,v]
 :=
 \int_\Om \nabla u\cdot \overline{\nabla v}\,\dd x
 +\int_{\pa\Om}\beta\,\tr{u}\,\overline{\tr{v}}\,\dd\sigma,
 \qquad \Dom(a_\beta)=H^1(\Om).
\end{equation*}
For $u\in H^1(\Omega)$ such that $\Delta u\in L^2(\Omega)$, we denote by
$\tr_\nu u\in H^{-1/2}(\partial\Omega)$ its weak exterior normal derivative,
characterized by Green's identity
\begin{equation*}
 \int_\Omega \nabla u\cdot\overline{\nabla v}\,\dd x
 =
 \int_\Omega (-\Delta u)\overline v\,\dd x
 +\langle \tr_\nu u,\tr v\rangle_{\partial\Omega},
 \qquad v\in H^1(\Omega).
\end{equation*}
See \cite[Section~2]{BehrndtRohleder} for the trace spaces and weak normal
derivative on possibly unbounded Lipschitz domains.

The trace theorem and the boundedness of $\beta$ imply that the boundary term
is infinitesimally form-bounded with respect to the Neumann form. Hence
$a_\beta$ is densely defined, symmetric, closed, and bounded from below; see,
for example, \cite{Daners,GesztesyMitrea,BehrndtRohleder}. The Neumann case
corresponds to $\beta\equiv0$. The associated self-adjoint operator satisfies
\begin{equation*}
 \Dom(H_\beta)
 =\Bigl\{u\in H^1(\Om): -\Delta u\in L^2(\Om),\ \tr_{\nu}{u}+\beta\tr{u}=0\quad\text{on}\,\partial\Omega\Bigr\},
\end{equation*}
where the boundary condition is understood in $H^{-1/2}(\pa\Om)$. The operator acts as
$ H_\beta u=-\Delta u.$

Since $H_\beta$ is bounded from below, we fix once and for all a real number
$ \lambda_\beta<\min\{0,\inf\sigma(H_\beta)\}$
so negative that, for every $\lambda\le\lambda_\beta$, the expression
$ a_\beta[u]-\lambda\|u\|_{L^2(\Om)}^2$
is positive and equivalent to the squared $H^1(\Om)$ norm. Such parameters
automatically belong to $\rho(H_\beta)$. We shall call them
\emph{admissible parameters}.

\subsection{Point evaluation on the graph domain}

Let $Y=\{y_1,\dots,y_N\}\subset\Om$ be a finite set of pairwise distinct points. Define
\begin{equation*}
 \bftau_Yu:=(u(y_1),\dots,u(y_N)),
 \qquad u\in\Dom(H_\beta).
\end{equation*}
If $U\Subset\Om$ is a neighbourhood of $Y$, interior elliptic regularity gives
\begin{equation*}
 \|u\|_{H^2(U)}\le C\bigl(\|u\|_{L^2(\Om)}+\|H_\beta u\|_{L^2(\Om)}\bigr),
 \qquad u\in\Dom(H_\beta).
\end{equation*}
Since $d\in\{2,3\}$, the embedding $H^2(U)\hookrightarrow C^0(\overline U)$
shows that $\bftau_Y:\Dom(H_\beta)\to\C^N$ is bounded with respect to the
graph norm of $H_\beta$.
This is the finite-dimensional trace map used in extension constructions for
point perturbations (see for example 
\cite[Section~1.4.3]{BrueningGeylerPankrashkin2008} and
\cite{Posilicano2008}).

\subsection{Free resolvent}
For 
$\lambda<0,
$
let
\begin{equation*}
\kappa(\lambda):=\sqrt{-\lambda}>0 .
\end{equation*}
The free resolvent kernel of $-\Delta$ on $\R^d$ is denoted by $G^0_{d,\lambda}(x,y)=G^0_{d,\lambda}(x-y)$ and is given by
\begin{equation}
\label{eq:free-resolvent-kernel}
G_{3, \lambda}^{0}(x,y)=\frac{e^{-\kappa(\lambda) |x-y|}}{4\pi |x-y|},
 \qquad
G_{2, \lambda}^{0}(x,y)=\frac1{2\pi}K_0(\kappa(\lambda) |x-y|),
\end{equation}
where $K_0$ denotes the modified Bessel function. Thus
\begin{equation}
\label{eq:3dasymptotics}
G_{3, \lambda}^{0}(x,y)=\frac1{4\pi |x-y|}-\frac{\kappa(\lambda)}{4\pi}+O(|x-y|),
 \qquad x\to y,
\end{equation}
and
\begin{equation}
\label{eq:2dasymptotics}
G_{2, \lambda}^{0}(x,y)
 =-\frac1{2\pi}\log |x-y|
 -\frac1{2\pi}\Bigl(\log\frac{\kappa(\lambda)}{2}+\gamma\Bigr)+o(1),
 \qquad x\to y,
\end{equation}
where $\gamma$ is Euler's constant.
See, for example, \cite{AGHH} and \cite[Section~1.1]{DerezinskiPoint}.
\subsection{Green kernels, regular parts and Weyl matrix function}
\begin{theorem}[Robin resolvent kernel]\label{thm:robin-kernel}
Let $\Omega$ and $\beta$ satisfy the standing assumptions and let $\lambda\le\lambda_\beta$. For every $y\in\Omega$ there exists a unique $h^{\beta}_{\lambda,y}\in H^1(\Omega)$ satisfying
\begin{equation*}
\begin{cases}
(-\Delta-\la)h^{\beta}_{\la,y}=0 & \text{in }\Omega,\\
\tr_{\nu}{h_{\lambda,y}^{\beta}}+\beta\tr{h_{\lambda,y}^{\beta}}=\tr_{\nu}{G_{d,\lambda}^{0}(\cdot,y)}+\beta\tr{G_{d,\lambda}^{0}(\cdot,y)} & \text{on }\pa\Omega.
\end{cases}
\end{equation*}

Put
\begin{equation*}
 G^\beta_\lambda(x,y)
 :=
 G^0_{d,\lambda}(x,y)-h^\beta_{\lambda,y}(x) \qquad \text{and}
 \qquad
 G^\beta_{\lambda,y}:=G^\beta_\lambda(\cdot,y) ,  
\end{equation*}
respectively the Green kernel and the Green function with pole at $y$. Then
$G^\beta_{\lambda,y}\in L^2(\Omega)$ and
\begin{equation*}
 (-\Delta-\lambda)G^\beta_{\lambda,y}=\delta_y
\end{equation*}
in the distributional sense, with the Robin boundary condition. Moreover,
\begin{equation}\label{eq:green-vector-duality} 
\bigl((H_\beta-\lambda)^{-1}f\bigr)(y)
 =
 (f,G^\beta_{\lambda,y})_{L^2(\Omega)}.
\end{equation}
Thus $G^\beta_\lambda(x,y)$ is the integral kernel of
$(H_\beta-\lambda)^{-1}$. It is Hermitian symmetric and smooth away
from the diagonal and has the same diagonal singularity as the free kernel.
\end{theorem}
Existence and uniqueness of the Robin correction follow from the boundary
solvability result \cite[Lemma~5.1]{BehrndtRohleder}. The remaining argument
parallels the analogous Dirichlet construction in
\cite{NojaRasoStoia2026}; for the universality of the diagonal singularity,
see \cite[Theorems~15--16]{BrueningGeylerPankrashkin2005}. Details are omitted.

\begin{definition}[Local charge and regular part]\label{def:local-regular-part}
Let $U\subset\Om$ be a neighbourhood of $y$. We denote by $\mathcal F_y(U)$
the class of functions, or distributions represented by functions on $U\setminus\{y\}$,
which have one of the following local expansions. In dimension three,
\begin{equation*}
 u(x)=\frac{q}{4\pi |x-y|}+u_{\reg}(x),
 \qquad u_{\reg}\in C^0(U),
\end{equation*}
whereas in dimension two,
\begin{equation*}
 u(x)=-\frac{q}{2\pi}\log |x-y|+u_{\reg}(x),
 \qquad u_{\reg}\in C^0(U).
\end{equation*}
The scalar $q$ is the charge of $u$ at $y$. The function $u_{\reg}$ is the
regular part of $u$ near $y$, and its value at the centre is
\begin{equation*}
 u_{\reg}(y):=\lim_{x\to y}
 \left(u(x)-\frac{q}{4\pi |x-y|}\right),
 \qquad d=3,
\end{equation*}
respectively
\begin{equation*}
 u_{\reg}(y):=\lim_{x\to y}
 \left(u(x)+\frac{q}{2\pi}\log |x-y|\right),
 \qquad d=2.
\end{equation*}
Here $u_{\reg}$ is part of the notation attached to the chosen singular expansion.
\end{definition}
The coefficient $q$ is uniquely determined by the leading asymptotics. If
$Y=\{y_1,\dots,y_N\}$, then $u\in\mathcal F_Y$ means that $u$ belongs locally
to $\mathcal F_{y_j}$ at each centre, and we write
\begin{equation*}
 u_{\reg,Y}:=\bigl(u_{\reg}(y_1),\dots,u_{\reg}(y_N)\bigr)
\end{equation*}
for the vector of regular-part values.
These charge and regular-part traces are the standard boundary data for point
interactions; see \cite{AGHH} and
\cite[Section~1.4.3]{BrueningGeylerPankrashkin2008}.

The regular-part value of $G^\beta_{\lambda,y}$ at the pole is
\begin{equation*}
 m^\beta_\lambda(y)
 :=(G^\beta_{\lambda,y})_{\reg}(y),
\end{equation*}
and \eqref{eq:3dasymptotics} gives
\begin{equation*}
 m^\beta_\lambda(y)
 =-\frac{\kappa(\lambda)}{4\pi}-h^\beta_{\lambda,y}(y),
 \qquad d=3,
\end{equation*}
and from \eqref{eq:2dasymptotics}
\begin{equation*}
 m^\beta_\lambda(y)
 =-\frac1{2\pi}\Bigl(\log\frac{\kappa(\lambda)}{2}+\gamma\Bigr)-h^\beta_{\lambda,y}(y),
 \qquad d=2.
\end{equation*}

For $Y=\{y_1,\dots,y_N\}$ define the renormalized Green matrix
\begin{equation}\label{eq:M-matrix}
 M^\beta_Y(\lambda):=\bigl(M^\beta_{jk}(\lambda)\bigr)_{j,k=1}^N
\end{equation}
by
\begin{equation*}
 M^\beta_{jj}(\lambda)=m^\beta_\lambda(y_j),
 \qquad
 M^\beta_{jk}(\lambda)=G^\beta_{\lambda,y_k}(y_j),\quad j\ne k.
\end{equation*}
The matrix is Hermitian for real $\lambda$ in the resolvent set.
The identity below is the point-interaction version of the standard
Weyl-function identity; compare
\cite[Definition~1.19 and Proposition~1.21]{BrueningGeylerPankrashkin2008}.

\begin{lemma}[Resolvent identity for regular parts]\label{lem:weyl-identity}
For $\lambda,\mu\le\lambda_\beta$ one has
\begin{equation}\label{eq:green-difference}
 G^\beta_{\lambda,y}-G^\beta_{\mu,y}
 =
 (\lambda-\mu)(H_\beta-\lambda)^{-1}G^\beta_{\mu,y}
 =
 (\lambda-\mu)(H_\beta-\mu)^{-1}G^\beta_{\lambda,y}.
\end{equation}
In particular, $G^\beta_{\lambda,y}-G^\beta_{\mu,y}\in\Dom(H_\beta)\subset H^1(\Om)$, and
\begin{equation}\label{eq:M-weyl-identity}
 M^\beta_Y(\lambda)-M^\beta_Y(\mu)
 =
 (\lambda-\mu)
 \bigl((G^\beta_{\mu,y_k},G^\beta_{\lambda,y_j})_{L^2(\Om)}\bigr)_{j,k=1}^N.
\end{equation}
\end{lemma}

\begin{proof}
The identities in \eqref{eq:green-difference} are the resolvent identity applied to the distributional source $\delta_y$, or equivalently to the Green function characterized by \eqref{eq:green-vector-duality}. Since $(H_\beta-\lambda)^{-1}$ maps $L^2(\Om)$ into $\Dom(H_\beta)$, the difference belongs to $\Dom(H_\beta)$.

The difference of two Green functions has no singularity at the pole, so its value at $y_j$ is the difference of the corresponding regular parts when $j=k$ and the ordinary difference of off-diagonal values when $j\ne k$. Evaluating \eqref{eq:green-difference} at $y_j$ and using \eqref{eq:green-vector-duality} gives \eqref{eq:M-weyl-identity}.
\end{proof}

For $z\in\rho(H_\beta)$, choose an admissible $\lambda$ and define
\begin{equation*}
 G^\beta_{z,y}
 :=
 G^\beta_{\lambda,y}
 +(z-\lambda)(H_\beta-z)^{-1}G^\beta_{\lambda,y},
 \qquad
 \gamma_\beta(z)q:=\sum_{j=1}^Nq_jG^\beta_{z,y_j}.
\end{equation*}
The definition is independent of the admissible parameter. Let
$M^\beta_Y(z)$ be the matrix of regular diagonal values and off-diagonal
values of these Green functions, as in \eqref{eq:M-matrix}. Then
$M^\beta_Y(z)^*=M^\beta_Y(\overline z)$, and the resolvent identity gives
\begin{equation}\label{eq:complex-weyl-identity}
 M^\beta_Y(z)-M^\beta_Y(w)^*
 =(z-\overline w)\gamma_\beta(w)^*\gamma_\beta(z),
 \qquad z,w\in\rho(H_\beta).
\end{equation}
In particular,
\begin{equation*}
 \operatorname{Im}M^\beta_Y(z)
 =
 (\operatorname{Im}z)\gamma_\beta(z)^*\gamma_\beta(z).
\end{equation*}
The map $\gamma_\beta(z)$ is injective because its components have distinct
singularities. Thus $M^\beta_Y$ is a strict matrix-valued Nevanlinna function.
For every real $E\in\rho(H_\beta)$, differentiation of
\eqref{eq:complex-weyl-identity} gives
\begin{equation}\label{eq:weyl-derivative}
 \frac{\dd}{\dd E}M^\beta_Y(E)
 =
 \bigl((G^\beta_{E,y_k},G^\beta_{E,y_j})_{L^2(\Omega)}\bigr)_{j,k=1}^N
 =\gamma_\beta(E)^*\gamma_\beta(E)>0.
\end{equation}

Accordingly, we shall call
$M_Y^\beta:\rho(H_\beta)\to\mathbb C^{N\times N}$ the Weyl matrix
function associated with $H_\beta$ and the interaction set $Y$; for fixed
$z\in\rho(H_\beta)$, $M_Y^\beta(z)$ is the corresponding Weyl matrix.
Equations \eqref{eq:complex-weyl-identity} and
\eqref{eq:weyl-derivative} are the standard Weyl-function identities (compare
\cite[Theorem~1.23]{BrueningGeylerPankrashkin2008} and
\cite[Definition~2.2 and Lemma~2.3(i)]{BehrndtLangerLotoreichikRohleder2017}).

The following lemma provides the large negative energy information needed for the
spectral analysis. Together with the resolvent criterion
\eqref{eq:resolvent-criterion}, it places a whole negative half-line in
$\rho(H^\beta_{\Theta,Y})$ and hence proves lower semiboundedness. Its
positivity statement is used in
Corollary~\ref{cor:multicenter-count}, while the leading terms in the asymptotic developments yield the
strong-coupling asymptotics in Corollary~\ref{cor:strong-coupling}.

\begin{lemma}[Large negative energies]\label{lem:large-negative}
As $\kappa\to+\infty$,
\begin{equation}\label{eq:large-negative-M}
 M^\beta_Y(-\kappa^2)
 =
 \begin{cases}
 -\dfrac{\kappa}{4\pi}I_N+o(1),& d=3,\\[6pt]
 -\dfrac1{2\pi}\left(\log\dfrac{\kappa}{2}+\gamma\right)I_N+o(1),& d=2,
 \end{cases}
\end{equation}
in matrix norm. Consequently, for every self-adjoint $\Theta$ there exists $\kappa_0>0$ such that $\Theta-M^\beta_Y(-\kappa^2)$ is positive definite for $\kappa\ge\kappa_0$.
\end{lemma}

\begin{proof}
The off-diagonal free Green functions tend to zero exponentially because
the centres are distinct. Put $\delta=\dist(Y,\partial\Omega)>0$. The explicit
three-dimensional kernel and the large-argument estimates for $K_0$ and $K_1$
in dimension two, see \cite[Appendix~B]{DerezinskiPoint}, imply that, for
$\kappa\ge1$ and $r=|x-y_j|\ge\delta$,
\begin{equation*}
 \left|
 (\tr_\nu+\beta\tr)G^0_{d,-\kappa^2}(x,y_j)
 \right|
 \le C(1+\kappa)e^{-\kappa r},
 \qquad x\in\partial\Omega.
\end{equation*}
If $\Omega$ is an exterior domain, compactness of $\partial\Omega$ gives the
corresponding $L^2(\partial\Omega)$ estimate directly. If $\Omega$ is a
special domain, write $x=(x',\varphi(x'))$ and $y_j=(y_j',y_{j,d})$. Then
\begin{equation*}
 r\ge\delta,
 \qquad
 r\ge|x'-y_j'|,
 \qquad
 2r\ge\delta+|x'-y_j'|.
\end{equation*}
Moreover, the boundedness of $\nabla\varphi$ implies $\dd\sigma\le C\,\dd x'$.
Consequently,
\begin{equation*}
\begin{aligned}
 &\left\|
 (\tr_\nu+\beta\tr)G^0_{d,-\kappa^2}(\cdot,y_j)
 \right\|_{L^2(\partial\Omega)}^2 \\
 &\qquad\le
 C(1+\kappa)^2e^{-\delta\kappa}
 \int_{\mathbb R^{d-1}}e^{-\kappa|x'-y_j'|}\,\dd x'
 \le C(1+\kappa)^2\kappa^{-(d-1)}e^{-\delta\kappa}.
\end{aligned}
\end{equation*}
Since $L^2(\partial\Omega)$ embeds continuously into
$H^{-1/2}(\partial\Omega)$ in the global trace framework of
\cite[Section~2]{BehrndtRohleder}, both classes of domains satisfy
\begin{equation}\label{eq:boundary-data-large-kappa}
 \bigl\|(\tr_\nu+\beta\tr)
 G^0_{d,-\kappa^2}(\cdot,y_j)\bigr\|_{H^{-1/2}(\partial\Omega)}
 \le C(1+\kappa)e^{-\delta\kappa/2},
 \qquad j=1,\ldots,N.
\end{equation}
The constants may be chosen uniformly in $j$ because $Y$ is finite. For sufficiently large $\kappa$, the trace inequality and the lower
semiboundedness of the Robin form imply the uniform coercive estimate
\begin{equation*}
 a_\beta[v]+\kappa^2\|v\|_{L^2(\Omega)}^2
 \ge c_0\|v\|_{H^1(\Omega)}^2.
\end{equation*}
The weak Robin problem and \eqref{eq:boundary-data-large-kappa} thus yield
\begin{equation*}
 \|h^\beta_{-\kappa^2,y_j}\|_{H^1(\Omega)}
 \le C'(1+\kappa)e^{-\delta\kappa/2}.
\end{equation*}
On neighbourhoods compactly contained in $\Omega$, interior elliptic
estimates for $(-\Delta+\kappa^2)h=0$, followed by
$H^2\hookrightarrow C^0$ for $d\le3$, cost only an additional polynomial
factor in $\kappa$. Hence
\begin{equation*}
 h^\beta_{-\kappa^2,y_k}(y_j)=o(1)
\end{equation*}
uniformly in $j,k$. Formula \eqref{eq:large-negative-M} now follows from
\eqref{eq:3dasymptotics}--\eqref{eq:2dasymptotics}.
\end{proof}

\section{Point-interaction realization}\label{sec:operator}

Let $\Theta=\Theta^*\in\mathbb C^{N\times N}$. After the preparation in the preceding Section, we will realize the point interactions as self-adjoint extensions of suitable restrictions, making use of an ordinary boundary triple; see
\cite[Corollary~1.27, Proposition~1.28, and
Theorem~1.29]{BrueningGeylerPankrashkin2008} for the abstract construction
and Kre\u{\i}n formula. To this end, let us consider the closed symmetric restriction
\begin{equation}\label{eq:minimal-point-restriction}
 S_{\beta,Y}
 :=
 H_\beta\upharpoonright
 \bigl\{u\in\Dom(H_\beta):\bftau_Yu=0\bigr\}.
\end{equation}
The graph-continuity of $\bftau_Y$ was established in
Section~\ref{sec:preliminaries}. The map is onto, as is seen by choosing
functions in $C^\infty_{\mathrm c}(\Omega)$ with arbitrarily prescribed
values on $Y$. Hence $S_{\beta,Y}$ has deficiency indices $(N,N)$. For any
admissible real $\lambda$, one has the direct decomposition
\begin{equation}\label{eq:adjoint-domain-decomposition}
 \Dom(S_{\beta,Y}^*)
 =
 \Dom(H_\beta)\dotplus\gamma_\beta(\lambda)\mathbb C^N.
\end{equation}
Thus every $\psi\in\Dom(S_{\beta,Y}^*)$ is represented uniquely as
\begin{equation}\label{eq:operator-domain-decomp}
 \psi=\phi+\gamma_\beta(\lambda)q,
 \qquad \phi\in\Dom(H_\beta),\quad q\in\mathbb C^N,
\end{equation}
and
\begin{equation}\label{eq:adjoint-action}
 (S_{\beta,Y}^*-\lambda)\psi=(H_\beta-\lambda)\phi.
\end{equation}
Define
\begin{equation}\label{eq:point-boundary-maps}
 \Gamma_0\psi:=q,
 \qquad
 \Gamma_1\psi:=\bftau_Y\phi+M_Y^\beta(\lambda)q.
\end{equation}
The second expression is precisely the vector $\psi_{\reg,Y}$ of regular
values at the centres. It follows from
\eqref{eq:green-difference}--\eqref{eq:M-weyl-identity} that the boundary
maps do not depend on the admissible parameter used in
\eqref{eq:operator-domain-decomp}.

\begin{theorem}[Boundary-triple realization]\label{thm:domain-action}
The triple
\begin{equation*}
 \bigl(\mathbb C^N,\Gamma_0,\Gamma_1\bigr)
\end{equation*}
is an ordinary boundary triple for $S_{\beta,Y}^*$. Its reference extension
$S_{\beta,Y}^*\upharpoonright\Ker\Gamma_0$ is $H_\beta$, and its gamma field
and Weyl function are exactly $\gamma_\beta$ and $M_Y^\beta$ introduced in
Section~\ref{sec:preliminaries}.

For every $\Theta=\Theta^*$, the operator
\begin{equation}\label{eq:point-operator-boundary-triple}
 H^\beta_{\Theta,Y}
 :=
 S_{\beta,Y}^*
 \upharpoonright\Ker(\Gamma_1-\Theta\Gamma_0)
\end{equation}
is self-adjoint and bounded from below. Equivalently,
$\psi\in\Dom(H^\beta_{\Theta,Y})$ if and only if it has the decomposition
\eqref{eq:operator-domain-decomp} and
\begin{equation}\label{eq:point-boundary-condition}
 \bftau_Y\phi=(\Theta-M_Y^\beta(\lambda))q.
\end{equation}
Its action is given by
\begin{equation*}
 (H^\beta_{\Theta,Y}-\lambda)\psi=(H_\beta-\lambda)\phi.
\end{equation*}
In local terms, \eqref{eq:point-boundary-condition} is the charge--regular
part condition
\begin{equation}\label{eq:renormalized-bc}
 \psi_{\reg,Y}=\Theta q.
\end{equation}
\end{theorem}

\begin{proof}
Green-function duality identifies
$\Ker(S_{\beta,Y}^*-z)$ with $\Ran\gamma_\beta(z)$. For
\begin{equation*}
 \psi=\phi+\gamma_\beta(\lambda)q,\qquad
 \eta=\chi+\gamma_\beta(\lambda)p,
\end{equation*}
\eqref{eq:green-vector-duality} gives
\begin{equation*}
 (S_{\beta,Y}^*\psi,\eta)-(\psi,S_{\beta,Y}^*\eta)
 =
 \langle\Gamma_1\psi,\Gamma_0\eta\rangle_{\mathbb C^N}
 -
 \langle\Gamma_0\psi,\Gamma_1\eta\rangle_{\mathbb C^N}.
\end{equation*}
The boundary map is onto: given $q,r\in\mathbb C^N$, choose
$\phi\in\Dom(H_\beta)$ such that
$\bftau_Y\phi=r-M_Y^\beta(\lambda)q$ and set
$\psi=\phi+\gamma_\beta(\lambda)q$. Then
$(\Gamma_0\psi,\Gamma_1\psi)=(q,r)$. The abstract construction cited above
therefore shows that these maps form an ordinary boundary triple. If
$\psi=\gamma_\beta(z)q$, then
\eqref{eq:green-difference} yields
$\Gamma_0\psi=q$ and $\Gamma_1\psi=M_Y^\beta(z)q$; hence the induced gamma
field and Weyl function are the stated ones.

The self-adjointness of \eqref{eq:point-operator-boundary-triple} is the
standard parametrization of self-adjoint extensions by Hermitian boundary
conditions. Equations \eqref{eq:point-boundary-condition} and
\eqref{eq:adjoint-action} give the domain and action, and
\eqref{eq:point-boundary-maps} gives \eqref{eq:renormalized-bc}. Finally,
Lemma~\ref{lem:large-negative} and the resolvent criterion below show that a
whole negative half-line belongs to $\rho(H^\beta_{\Theta,Y})$, proving
lower semiboundedness.
\end{proof}
\begin{remark}
A diagonal matrix
$\Theta=\diag(\alpha_1,\ldots,\alpha_N)$ describes local point
interactions with strengths $\alpha_j$, while a non-diagonal matrix also
couples the charges at distinct centres. \end{remark}
\begin{corollary}[Kre\u{\i}n formula and spectral consequences]\label{cor:krein-form}
For every $z\in\rho(H_\beta)$,
\begin{equation}\label{eq:resolvent-criterion}
 z\in\rho(H^\beta_{\Theta,Y})
 \quad\Longleftrightarrow\quad
 0\in\rho\bigl(\Theta-M^\beta_Y(z)\bigr).
\end{equation}
When these conditions hold,
\begin{equation}\label{eq:krein-form}
 (H^\beta_{\Theta,Y}-z)^{-1}
 =
 (H_\beta-z)^{-1}
 +
 \gamma_\beta(z)\bigl(\Theta-M^\beta_Y(z)\bigr)^{-1}
 \gamma_\beta(\overline z)^*.
\end{equation}
For $E\in\rho(H_\beta)\cap\R$,
\begin{equation}\label{eq:eigenvalue-criterion}
 \Ker(H^\beta_{\Theta,Y}-E)
 =
 \gamma_\beta(E)\Ker\bigl(\Theta-M^\beta_Y(E)\bigr),
\end{equation}
and the two kernels have the same dimension. Moreover,
\begin{equation}\label{eq:essential-spectrum}
 \sigma_{\mathrm{ess}}(H^\beta_{\Theta,Y})
 =
 \sigma_{\mathrm{ess}}(H_\beta).
\end{equation}
\end{corollary}

\begin{proof}
The resolvent and kernel identities follow from the ordinary
boundary-triple Kre\u{\i}n formula
\cite[Theorem~1.29]{BrueningGeylerPankrashkin2008}. Since
$\gamma_\beta(z)$ maps $\mathbb C^N$ into $L^2(\Omega)$, the resolvent
difference has rank at most $N$. Weyl's theorem therefore gives
\eqref{eq:essential-spectrum}.
\end{proof}

\section{Subthreshold spectrum and coupling asymptotics}
\label{sec:spectrum}

The kernel identity \eqref{eq:eigenvalue-criterion} detects
eigenvalues at a prescribed energy. Combined with the strict energy
monotonicity \eqref{eq:weyl-derivative} of the Weyl matrix, it yields the
following global counting principle.

\begin{corollary}[Eigenvalue count below the background spectrum]\label{cor:multicenter-count}
For every $E_0<\inf\sigma(H_\beta)$,
\begin{equation}\label{eq:multicenter-count-below-energy}
 \dim\Ran\mathbf{1}_{(-\infty,E_0)}(H^\beta_{\Theta,Y})
 =n_-\bigl(\Theta-M^\beta_Y(E_0)\bigr),
\end{equation}
where $n_-(A)$ is the number of negative eigenvalues of a Hermitian matrix
$A$, counted with multiplicity. In particular, $H^\beta_{\Theta,Y}$ has at
most $N$ eigenvalues below $\inf\sigma(H_\beta)$.

Assume in addition that $H_\beta\ge0$. Then
\begin{equation}\label{eq:multicenter-negative-count-limit}
 \dim\Ran\mathbf{1}_{(-\infty,0)}(H^\beta_{\Theta,Y})
 =
 \lim_{\kappa\downarrow0}
 n_-\bigl(\Theta-M^\beta_Y(-\kappa^2)\bigr).
\end{equation}
The limit on the right-hand side always exists. If
$M^\beta_Y(-\kappa^2)$ converges in $\mathbb C^{N\times N}$ as
$\kappa\downarrow0$, denote its finite Hermitian limit by
\begin{equation*}
 M^\beta_Y(0):=\lim_{\kappa\downarrow0}M^\beta_Y(-\kappa^2).
\end{equation*}
Then
\begin{equation}\label{eq:multicenter-negative-count}
 \dim\Ran\mathbf{1}_{(-\infty,0)}(H^\beta_{\Theta,Y})
 =n_-\bigl(\Theta-M^\beta_Y(0)\bigr).
\end{equation}
In particular, if $\Theta=\alpha I_N$ and
$
 \mu_1\le\cdots\le\mu_N
$
are the eigenvalues of $M^\beta_Y(0)$, then
\begin{equation*}
 \dim\Ran\mathbf{1}_{(-\infty,0)}(H^\beta_{\alpha I_N,Y})
 =\#\{j:\mu_j>\alpha\}.
\end{equation*}
Thus the first negative eigenvalue appears when $\alpha<\mu_N$, while
further eigenvalues appear as $\alpha$ crosses the remaining eigenvalues
of $M^\beta_Y(0)$.
\end{corollary}

\begin{proof}
Set $A(E)=\Theta-M^\beta_Y(E)$ for
$E<\inf\sigma(H_\beta)$, and denote its ordered eigenvalues by
\begin{equation*}
 a_1(E)\le\cdots\le a_N(E).
\end{equation*}
By Lemma~\ref{lem:large-negative}, all $a_j(E)$ are positive for
sufficiently negative $E$. If
$E_1<E_2<\inf\sigma(H_\beta)$, \eqref{eq:weyl-derivative} gives
\begin{equation*}
 A(E_1)-A(E_2)=M^\beta_Y(E_2)-M^\beta_Y(E_1)>0.
\end{equation*}
The finite-dimensional min--max principle therefore implies
$a_j(E_1)>a_j(E_2)$ for every $j$. By
\eqref{eq:eigenvalue-criterion}, the zeros of the $a_j$, counted with
their null multiplicities, are precisely the eigenvalues of
$H^\beta_{\Theta,Y}$ below the background spectrum. Each $a_j$ crosses
zero at most once. The number of crossings in $(-\infty,E_0)$ is therefore
$n_-(A(E_0))$; a zero eigenvalue of $A(E_0)$ is correctly excluded by the
open interval. This proves \eqref{eq:multicenter-count-below-energy}.
If $H_\beta\ge0$, take $E_0=-\kappa^2$ and let $\kappa\downarrow0$.
The spectral projections of $H^\beta_{\Theta,Y}$ associated with
$(-\infty,-\kappa^2)$ increase to the spectral projection associated with
$(-\infty,0)$. Since their ranks are bounded by $N$,
\eqref{eq:multicenter-count-below-energy} gives
\eqref{eq:multicenter-negative-count-limit}; in particular, the integer-valued
limit on its right-hand side exists. If $M^\beta_Y(-\kappa^2)$ converges to
the finite matrix $M^\beta_Y(0)$, the ordered eigenvalues of
$A(-\kappa^2)$ converge to those of $A(0)$. Moreover,
$A(-\kappa^2)>A(0)$ for every $\kappa>0$, so an eigenvalue that converges
to zero is positive before the limit and is correctly excluded. It follows
that the limit in \eqref{eq:multicenter-negative-count-limit} equals
$n_-(A(0))$, proving \eqref{eq:multicenter-negative-count}.
\end{proof}

Thus the discrete spectrum below the background spectrum is
controlled entirely by the finite-dimensional inertia of
$\Theta-M_Y^\beta(E)$. In particular, the point perturbation can produce at
most $N$ eigenvalues below $\inf\sigma(H_\beta)$. When $H_\beta\geq0$, the
limiting formula at zero counts all negative eigenvalues without requiring
the existence of a finite zero-energy Weyl matrix. If $M_Y^\beta(0)$ exists,
its eigenvalues are the successive critical values of the scalar coupling
$\alpha$ at which the number of negative eigenvalues changes. These statements
concern the spectrum strictly below zero; they do not determine whether a
zero-energy eigenfunction or resonance occurs at a critical coupling.

For $N>1$ and $E\in\rho(H_\beta)\cap\mathbb R$, the eigenvalue
condition is
$ \det\bigl(\Theta-M_Y^\beta(E)\bigr)=0,$
and the multiplicity is
$\dim\Ker\bigl(\Theta-M_Y^\beta(E)\bigr)$.
We next refine this general multicentre picture. We first give
the complete description of the eigenvalue branch for one centre and then
consider the strong-coupling regime for scalar multicentre interactions.
In the one-centre case $\Theta=\alpha\in\R$, the local condition is
$\psi_{\reg}(y)=\alpha q$ and \eqref{eq:krein-form} reduces, at real
$z\in\rho(H_\beta)\cap\rho(H^\beta_{\alpha,y})$, to
\begin{equation*}
 (H^\beta_{\alpha,y}-z)^{-1}
 =
 (H_\beta-z)^{-1}
 +
 \frac{|G^\beta_{z,y}\rangle\langle G^\beta_{z,y}|}
 {\alpha-m^\beta_z(y)}.
\end{equation*}
The formal value $\alpha=\infty$ means $q=0$ and yields the unperturbed operator $H_\beta$.

\begin{proposition}[Single-centre spectrum below the background spectrum]
\label{prop:single-center-spectrum}
Set
\begin{equation*}
 I_\beta:=(-\infty,\inf\sigma(H_\beta)).
\end{equation*}
For one centre, the function $E\mapsto m^\beta_E(y)$ is real analytic and
strictly increasing on $I_\beta$. The extended limit
\begin{equation*}
 \alpha_*^\beta(y)
 :=\lim_{E\uparrow\inf\sigma(H_\beta)}m^\beta_E(y)
 \in\R\cup\{+\infty\}
\end{equation*}
exists, and the range of $m^\beta_E(y)$ on $I_\beta$ is
$(-\infty,\alpha_*^\beta(y))$. The operator $H^\beta_{\alpha,y}$ has an
eigenvalue in $I_\beta$ if and only if
$\alpha<\alpha_*^\beta(y)$. In that case the eigenvalue $E(\alpha)$ is
unique and simple, the map
\begin{equation*}
 E:(-\infty,\alpha_*^\beta(y))\longrightarrow I_\beta
\end{equation*}
is real analytic and strictly increasing, and a normalized eigenfunction is
$
 \psi_\alpha
 =
 \frac{G^\beta_{E(\alpha),y}}
 {\|G^\beta_{E(\alpha),y}\|_{L^2(\Omega)}}.
 $
\end{proposition}

\begin{proof}
The eigenvalue equation is
 $\alpha=m^\beta_{E(\alpha)}(y).$
The analyticity of the resolvent implies that of $m^\beta_E(y)$ on
$I_\beta$. By Lemma~\ref{lem:large-negative},
$m^\beta_E(y)\to-\infty$ as $E\to-\infty$, while monotonicity gives the
extended limit $\alpha_*^\beta(y)$ at the other endpoint. Hence the stated
range follows. Formula \eqref{eq:weyl-derivative} allows the
inverse-function theorem to be applied at every $E\in I_\beta$. The kernel formula
\eqref{eq:eigenvalue-criterion} shows that the resulting eigenvalue is
simple and gives the stated eigenfunction. Since the inverse of a strictly
increasing function is strictly increasing, the same holds for $E(\alpha)$.
\end{proof}

\begin{corollary}[Derivative and strict concavity]
\label{cor:one-center-concavity}
The one-centre eigenvalue branch satisfies
\begin{equation}\label{eq:one-center-coupling-derivative}
 \frac{\dd E}{\dd\alpha}
 =
 \frac1{\|G^\beta_{E(\alpha),y}\|_{L^2(\Omega)}^2}>0.
\end{equation}
Moreover, $E$ is strictly concave on
$(-\infty,\alpha_*^\beta(y))$, and
\begin{equation}\label{eq:one-center-second-derivative}
 \frac{\dd^2E}{\dd\alpha^2}
 =
 -\frac{
  2\bigl((H_\beta-E(\alpha))^{-1}G^\beta_{E(\alpha),y},
  G^\beta_{E(\alpha),y}\bigr)_{L^2(\Omega)}
 }{
  \|G^\beta_{E(\alpha),y}\|_{L^2(\Omega)}^6
 }
 <0.
\end{equation}
\end{corollary}

\begin{proof}
Differentiating
$\alpha=m^\beta_{E(\alpha)}(y)$ and using
\eqref{eq:weyl-derivative} gives
\begin{equation*}
 1
 =
 \|G^\beta_{E(\alpha),y}\|_{L^2(\Omega)}^2E'(\alpha),
\end{equation*}
which proves \eqref{eq:one-center-coupling-derivative}. The
Green-function resolvent identity also gives
\begin{equation*}
 \frac{\dd}{\dd E}G^\beta_{E,y}
 =
 (H_\beta-E)^{-1}G^\beta_{E,y},
\end{equation*}
and consequently
\begin{equation*}
 \frac{\dd^2}{\dd E^2}m^\beta_E(y)
 =
 2\bigl((H_\beta-E)^{-1}G^\beta_{E,y},
 G^\beta_{E,y}\bigr)_{L^2(\Omega)}>0.
\end{equation*}
The second derivative of the inverse function is
\begin{equation*}
 E''(\alpha)
 =
 -\frac{(m^\beta_E(y))''}
 {\bigl((m^\beta_E(y))'\bigr)^3}
 \bigg|_{E=E(\alpha)},
\end{equation*}
which is \eqref{eq:one-center-second-derivative}.
\end{proof}

\begin{remark}[Charge monotonicity]
For the normalized eigenfunction, the charge satisfies
\begin{equation*}
 |q_\alpha|^2=E'(\alpha),
 \qquad
 \frac{\dd}{\dd\alpha}|q_\alpha|^2=E''(\alpha)<0.
\end{equation*}
Thus the singular charge decreases as the interaction becomes less
attractive; equivalently, the binding energy
$\inf\sigma(H_\beta)-E(\alpha)$ is strictly convex.
\end{remark}

\begin{corollary}[Universal multicentre strong-coupling asymptotics]
\label{cor:strong-coupling}
Let $\Theta=\alpha I_N$. For all sufficiently negative $\alpha$, the operator
$H^\beta_{\alpha I_N,Y}$ has exactly $N$ eigenvalues below
$\inf\sigma(H_\beta)$, counted with multiplicity. Denote them by
\begin{equation*}
 E_1(\alpha)\leq\cdots\leq E_N(\alpha).
\end{equation*}
Then, for every $j=1,\ldots,N$, as $\alpha\to-\infty$,
\begin{equation}\label{eq:strong-coupling-asymptotics}
 E_j(\alpha)
 =
 \begin{cases}
 -16\pi^2\alpha^2+o(|\alpha|),&d=3,\\[4pt]
 -4e^{-2\gamma}e^{-4\pi\alpha}\bigl(1+o(1)\bigr),&d=2.
 \end{cases}
\end{equation}
The remainder estimates are uniform in $j$. In particular, the leading terms
are independent of the domain, the interaction centres, and the Robin
coefficient, and coincide with the whole-space one-centre laws; see
\cite{AGHH,DerezinskiPoint}.
\end{corollary}

\begin{proof}
For $E<\inf\sigma(H_\beta)$, let
\begin{equation*}
 \mu_1(E)\leq\cdots\leq\mu_N(E)
\end{equation*}
be the ordered eigenvalues of the Hermitian matrix $M_Y^\beta(E)$.
Lemma~\ref{lem:large-negative} implies
\begin{equation}\label{eq:weyl-eigenvalues-large-negative}
 \mu_j(-\kappa^2)
 =
 \begin{cases}
 -\dfrac{\kappa}{4\pi}+o(1),&d=3,\\[6pt]
 -\dfrac{1}{2\pi}\left(\log\dfrac{\kappa}{2}+\gamma\right)+o(1),&d=2,
 \end{cases}
 \qquad \kappa\to+\infty,
\end{equation}
uniformly in $j$.

Fix $E_0<\min\{0,\inf\sigma(H_\beta)\}$. For all sufficiently negative $\alpha$,
\begin{equation*}
 n_-\bigl(\alpha I_N-M_Y^\beta(E_0)\bigr)=N,
\end{equation*}
because $\alpha<\mu_1(E_0)$. Corollary~\ref{cor:multicenter-count} therefore
shows that $H^\beta_{\alpha I_N,Y}$ has exactly $N$ eigenvalues below $E_0$,
counted with multiplicity. Since the same corollary bounds by $N$ the total
number of eigenvalues below $\inf\sigma(H_\beta)$, these are all the
eigenvalues below the background spectrum.

Write an arbitrary one of them as
$E_j(\alpha)=-\kappa_j(\alpha)^2$. Then
$\kappa_j(\alpha)\to+\infty$ as $\alpha\to-\infty$. Indeed, otherwise a
subsequence of $E_j(\alpha)$ would remain in a compact subset of
$(-\infty,E_0]$, on which $M_Y^\beta(E)$ is bounded, contradicting
$\alpha\in\sigma(M_Y^\beta(E_j(\alpha)))$. The eigenvalue criterion together
with \eqref{eq:weyl-eigenvalues-large-negative} now yields,
for some eigenvalue of $M_Y^\beta(-\kappa_j(\alpha)^2)$,
\begin{equation*}
 \alpha
 =
 \begin{cases}
 -\dfrac{\kappa_j(\alpha)}{4\pi}+o(1),&d=3,\\[6pt]
 -\dfrac{1}{2\pi}\left(\log\dfrac{\kappa_j(\alpha)}{2}+\gamma\right)
 +o(1),&d=2.
 \end{cases}
\end{equation*}
The error is uniform because $N$ is finite and the remainder in
Lemma~\ref{lem:large-negative} converges in matrix norm. Consequently,
\begin{equation*}
 \kappa_j(\alpha)=-4\pi\alpha+o(1),
 \qquad d=3,
\end{equation*}
whereas
\begin{equation*}
 \log\kappa_j(\alpha)
 =-2\pi\alpha+\log2-\gamma+o(1),
 \qquad d=2.
\end{equation*}
Since $E_j(\alpha)=-\kappa_j(\alpha)^2$, these relations give
\eqref{eq:strong-coupling-asymptotics}.
\end{proof}

\begin{corollary}[Isolated background ground state]
\label{cor:isolated-background-ground-state}
Assume that
\begin{equation*}
 E_\beta:=\inf\sigma(H_\beta)
\end{equation*}
is an isolated eigenvalue of $H_\beta$. Then $E_\beta$ is simple. Let
$\phi_\beta$ denote its normalized eigenfunction, chosen strictly positive in
$\Omega$. Then
\begin{equation*}
 \alpha_*^\beta(y)=+\infty.
\end{equation*}
Consequently, for every $\alpha\in\mathbb R$ the operator
$H^\beta_{\alpha,y}$ has a unique simple eigenvalue
$E(\alpha)<E_\beta$. Moreover,
\begin{equation}\label{eq:isolated-ground-state-decoupling}
 E_\beta-E(\alpha)
 =
 \frac{|\phi_\beta(y)|^2}{\alpha}
 +O(\alpha^{-2}),
 \qquad \alpha\to+\infty,
\end{equation}
and, for every $z\in\mathbb C\setminus\mathbb R$,
\begin{equation}\label{eq:norm-resolvent-decoupling}
 \lim_{\alpha\to+\infty}
 \bigl\|
 (H^\beta_{\alpha,y}-z)^{-1}-(H_\beta-z)^{-1}
 \bigr\|=0.
\end{equation}
\end{corollary}

\begin{proof} We prove the statements about the unperturbed operator $H_\beta$, because we do not have a precise reference justifying them.
Choose $c>0$ so large that the shifted form
$ a_{\beta,c}[u,v]
 :=a_\beta[u,v]+c(u,v)_{L^2(\Omega)}$
is non-negative. Since $|u|\in H^1(\Omega)$,
$|\nabla |u||\leq|\nabla u|$ almost everywhere, and
$\tr|u|=|\tr u|$, one has
\begin{equation*}
 a_{\beta,c}[|u|]\leq a_{\beta,c}[u],
 \qquad u\in H^1(\Omega).
\end{equation*}
The first Beurling--Deny criterion
(see \cite[Theorem~2.6]{Ouhabaz}) therefore shows that
$e^{-t(H_\beta+c)}$ is positivity preserving. Since
$e^{-t(H_\beta+c)}=e^{-ct}e^{-tH_\beta}$, the same is true of
$e^{-tH_\beta}$, and the two semigroups have the same invariant ideals. We verify irreducibility directly (notice that boundedness of $\Omega$ is not needed here).
Suppose that, for some measurable set $A\subset\Omega$, the ideal $L^2(A)$
is invariant under $e^{-tH_\beta}$. The form invariance criterion
\cite[Theorem~2.2]{Ouhabaz}, applied to the orthogonal projection given by
multiplication by $\mathbf 1_A$, implies that
\begin{equation}\label{eq:invariant-ideal-form-domain}
 \mathbf 1_Au\in H^1(\Omega)
 \qquad\text{for every }u\in H^1(\Omega).
\end{equation}
If $U\Subset\Omega$, choose $\eta\in C^\infty_{\mathrm c}(\Omega)$ such
that $\eta=1$ on a neighbourhood of $\overline U$. Taking $u=\eta$ in
\eqref{eq:invariant-ideal-form-domain} shows that
$\mathbf 1_A\in H^1(U)$; hence
$\mathbf 1_A\in H^1_{\mathrm{loc}}(\Omega)$. The Sobolev product rule and
$\mathbf 1_A^2=\mathbf 1_A$ give
\begin{equation*}
 (2\mathbf 1_A-1)\nabla\mathbf 1_A=0
 \qquad\text{almost everywhere in }\Omega.
\end{equation*}
Because $|2\mathbf 1_A-1|=1$ almost everywhere, it follows that
$\nabla\mathbf 1_A=0$. Connectedness of $\Omega$ then implies that
$\mathbf 1_A$ is constant almost everywhere. Thus either $A$ or
$\Omega\setminus A$ has measure zero, so the semigroup is irreducible.
For a positivity-preserving self-adjoint semigroup, irreducibility is
equivalent to
positivity improvement; see \cite[Corollary~2.11]{Ouhabaz}.
Since $E_\beta=\inf\sigma(H_\beta)$ is an eigenvalue, positivity improvement
implies that it is simple and that its normalized eigenfunction $\phi_\beta$
may be chosen positive almost everywhere; see
\cite[Theorem~XIII.44]{ReedSimonIV}. Interior elliptic regularity gives a
continuous representative of $\phi_\beta$, and the local Harnack inequality
then yields
\begin{equation*}
 \phi_\beta(x)>0 \qquad\text{for every }x\in\Omega.
\end{equation*}
In particular, $\phi_\beta(y)\ne0$.

We next identify the spectral pole seen by the Weyl function. Let $P_\beta$
be the orthogonal projection onto $\operatorname{span}\{\phi_\beta\}$,
that is,
\begin{equation*}
 P_\beta f=(f,\phi_\beta)_{L^2(\Omega)}\phi_\beta,
 \qquad Q_\beta:=I-P_\beta.
\end{equation*}
Since $E_\beta$ is isolated, the reduced resolvent
\begin{equation*}
 R_\beta^\perp(E)
 :=Q_\beta(H_\beta-E)^{-1}Q_\beta
\end{equation*}
extends analytically across $E_\beta$, and
\begin{equation}\label{eq:background-resolvent-pole}
 (H_\beta-E)^{-1}
 =\frac{P_\beta}{E_\beta-E}+R_\beta^\perp(E)
\end{equation}
for $E<E_\beta$ sufficiently close to $E_\beta$.
Fix $\mu<E_\beta$. The Green-function duality
\eqref{eq:green-vector-duality} and
$(H_\beta-\mu)^{-1}\phi_\beta=(E_\beta-\mu)^{-1}\phi_\beta$
give
\begin{equation}\label{eq:ground-state-green-pairing}
 (\phi_\beta,G^\beta_{\mu,y})_{L^2(\Omega)}
 =\frac{\phi_\beta(y)}{E_\beta-\mu}.
\end{equation}
Applying $P_\beta$ to the Green-function identity
\eqref{eq:green-difference}, and using
\eqref{eq:ground-state-green-pairing} (recalling the convention that the
scalar product is linear in its first argument), yields
\begin{equation*}
 P_\beta G^\beta_{E,y}
 =\frac{\overline{\phi_\beta(y)}}{E_\beta-E}\,\phi_\beta.
\end{equation*}
Consequently,
\begin{equation}\label{eq:green-function-ground-state-pole}
 G^\beta_{E,y}
 =\frac{\overline{\phi_\beta(y)}}{E_\beta-E}\,\phi_\beta
  +G^\perp_{E,y},
 \qquad G^\perp_{E,y}:=Q_\beta G^\beta_{E,y}.
\end{equation}
More explicitly,
\begin{equation*}
 G^\perp_{E,y}
 =Q_\beta G^\beta_{\mu,y}
 +(E-\mu)R_\beta^\perp(E)Q_\beta G^\beta_{\mu,y},
\end{equation*}
so $E\mapsto G^\perp_{E,y}$ extends analytically, as an $L^2(\Omega)$-valued
function, across $E_\beta$.

The first term on the right-hand side of
\eqref{eq:green-function-ground-state-pole} is smooth near $y$ and has zero
charge. Hence the canonical spatial singularity of $G^\beta_{E,y}$ remains
entirely in $G^\perp_{E,y}$: as $x\to y$,
\begin{equation*}
 G^\perp_{E,y}(x)
 =\begin{cases}
 \dfrac{1}{4\pi|x-y|}+O(1),&d=3,\\[6pt]
 -\dfrac{1}{2\pi}\log|x-y|+O(1),&d=2.
 \end{cases}
\end{equation*}
Equivalently, in the distributional sense,
\begin{equation*}
 (-\Delta-E)G^\perp_{E,y}
 =\delta_y-\overline{\phi_\beta(y)}\phi_\beta.
\end{equation*}
Thus regularization at $y$ according to
Definition~\ref{def:local-regular-part} removes the spatial singularity from
$G^\perp_{E,y}$ but leaves unchanged the regular, energy-dependent pole term
in \eqref{eq:green-function-ground-state-pole}. By linearity of the
regular-part trace,
\begin{equation}\label{eq:weyl-pole-background-ground-state}
 m^\beta_E(y)
 =\frac{|\phi_\beta(y)|^2}{E_\beta-E}+r_\beta(E;y),
 \qquad
 r_\beta(E;y):=(G^\perp_{E,y})_{\reg}(y).
\end{equation}

We show directly the real analyticity of the remainder, which does not follow
merely from applying point evaluation to an $L^2$-analytic family. The scalar
case of the Weyl identity \eqref{eq:M-weyl-identity} and
\eqref{eq:green-function-ground-state-pole} give
\begin{equation*}
\begin{aligned}
 m^\beta_E(y)-m^\beta_\mu(y)
 =(E-\mu)(G^\beta_{\mu,y},G^\beta_{E,y})_{L^2(\Omega)}
 =\frac{(E-\mu)|\phi_\beta(y)|^2}
         {(E_\beta-\mu)(E_\beta-E)}
   +(E-\mu)(G^\beta_{\mu,y},G^\perp_{E,y})_{L^2(\Omega)}.
\end{aligned}
\end{equation*}
Since
\begin{equation*}
 \frac{E-\mu}{(E_\beta-\mu)(E_\beta-E)}
 =\frac{1}{E_\beta-E}-\frac{1}{E_\beta-\mu},
\end{equation*}
the remainder in \eqref{eq:weyl-pole-background-ground-state} is
\begin{equation}\label{eq:weyl-pole-regular-remainder}
\begin{aligned}
 r_\beta(E;y)
 ={}m^\beta_\mu(y)
 -\frac{|\phi_\beta(y)|^2}{E_\beta-\mu}
 +(E-\mu)(G^\beta_{\mu,y},G^\perp_{E,y})_{L^2(\Omega)},
\end{aligned}
\end{equation}
which is real analytic near $E_\beta$.

Because $\phi_\beta(y)\ne0$, the coefficient of the pole in
\eqref{eq:weyl-pole-background-ground-state} is strictly positive. Hence
\begin{equation*}
 \lim_{E\uparrow E_\beta}m^\beta_E(y)=+\infty
 \qquad\text{and}\qquad
 \alpha_*^\beta(y)=+\infty.
\end{equation*}
Proposition~\ref{prop:single-center-spectrum} now shows that, for every
$\alpha\in\mathbb R$, $H^\beta_{\alpha,y}$ has a unique simple eigenvalue
$E(\alpha)<E_\beta$, determined by
\begin{equation}\label{eq:isolated-ground-state-eigenvalue-equation}
 \alpha=m^\beta_{E(\alpha)}(y).
\end{equation}
Since the Weyl function is strictly increasing and diverges at $E_\beta$,
its inverse satisfies
\begin{equation*}
 E(\alpha)\uparrow E_\beta
 \qquad\text{as }\alpha\to+\infty.
\end{equation*}

We make the inversion at infinity explicit. Set
\begin{equation*}
 A_y:=|\phi_\beta(y)|^2>0,
 \qquad s:=\alpha^{-1},
 \qquad \delta:=E_\beta-E.
\end{equation*}
For $s>0$ small, the eigenvalue equation is equivalent to
\begin{equation}\label{eq:implicit-decoupling-equation}
 F(s,\delta):=
 \delta-sA_y-s\delta\,r_\beta(E_\beta-\delta;y)=0.
\end{equation}
The function $F$ is real analytic near $(0,0)$ and satisfies
 $F(0,0)=0$ and
 $ \partial_\delta F(0,0)=1.$
The analytic implicit-function theorem therefore provides a unique analytic
solution $\delta=\delta(s)$ with $\delta(0)=0$. For $s>0$ this is the branch
$\delta(s)=E_\beta-E(1/s)$, by uniqueness of the eigenvalue. From
\eqref{eq:implicit-decoupling-equation},
\begin{equation*}
 \delta(s)
 =sA_y+s\delta(s)r_\beta(E_\beta-\delta(s);y),
\end{equation*}
and differentiation at $s=0$ gives $\delta'(0)=A_y$. Thus
$ \delta(s)=A_ys+O(s^2).$
Returning to $s=\alpha^{-1}$ proves
\eqref{eq:isolated-ground-state-decoupling}. In fact, writing
$r_0:=r_\beta(E_\beta;y)$ and substituting the first-order expansion into
\eqref{eq:implicit-decoupling-equation} gives the sharper relation
\begin{equation*}
 E_\beta-E(\alpha)
 =\frac{A_y}{\alpha}+\frac{A_yr_0}{\alpha^2}+O(\alpha^{-3}).
\end{equation*}

Finally, the norm-resolvent statement follows independently from the
point-interaction resolvent formula. For fixed
$z\in\mathbb C\setminus\mathbb R$, the scalar Kre\u{\i}n formula
\eqref{eq:krein-form} gives
\begin{equation*}
 (H^\beta_{\alpha,y}-z)^{-1}-(H_\beta-z)^{-1}
 =\gamma_\beta(z)\frac{1}{\alpha-m^\beta_z(y)}
  \gamma_\beta(\overline z)^*.
\end{equation*}
All quantities except $\alpha$ on the right-hand side are fixed. In
particular, $|\alpha-m^\beta_z(y)|\geq\alpha/2$ for all sufficiently large
$\alpha$, and hence
\begin{equation*}
 \bigl\|
 (H^\beta_{\alpha,y}-z)^{-1}-(H_\beta-z)^{-1}
 \bigr\|
 \leq
 \frac{\|\gamma_\beta(z)\|\,
       \|\gamma_\beta(\overline z)\|}
      {|\alpha-m^\beta_z(y)|}
 =O(\alpha^{-1}).
\end{equation*}
This proves \eqref{eq:norm-resolvent-decoupling}.
\end{proof}

\begin{remark}[Comparison with the bounded Dirichlet case]
\label{rem:comparison-lotoreichik-michelangeli}
Corollary~\ref{cor:isolated-background-ground-state} overlaps, in the bounded
Dirichlet setting, with some of the results proved in \cite{LotoMiche21}. More
precisely, for a bounded connected $C^\infty$ domain they prove that, for
every finite $\alpha$, there is exactly one simple eigenvalue below the first
Dirichlet eigenvalue and that this eigenvalue converges to the latter as
$\alpha\to+\infty$; see
\cite[Proposition~IV.1(i)--(iv)]{LotoMiche21}. The authors also prove convergence of
the point-interaction operators to the Dirichlet Laplacian in the strong
resolvent sense \cite[Proposition~III.3(v)]{LotoMiche21}. Their decoupling
argument uses monotone convergence of quadratic forms, whereas the present
proof is based on the pole of the background resolvent and the scalar Weyl
equation. It applies to the present unbounded Neumann and Robin setting
whenever the background spectral bottom is isolated, and yields both the
quantitative expansion \eqref{eq:isolated-ground-state-decoupling} and the
stronger norm-resolvent convergence
\eqref{eq:norm-resolvent-decoupling}.
\end{remark}
\section{Explicit radial models}\label{sec:models}

We conclude with two one-centre exterior models for which the zero-energy
Weyl limit, and hence the critical coupling, can be computed explicitly.
The Robin parameter is constant.

Let $\Omega_R=\{x\in\R^3:|x|>R\}$, $|y|=\rho>R$, and take a
constant Robin parameter $\beta\ge0$ on $|x|=R$. The zero-energy critical
value is
\begin{equation*}
 \alpha_c^{\beta,\Omega_R}(\rho)=
 \frac1{4\pi}
 \sum_{\ell=0}^{\infty}
 \frac{\ell-\beta R}{\ell+1+\beta R}
 \frac{R^{2\ell+1}}{\rho^{2\ell+2}} .
\end{equation*}
In particular,
\begin{equation*}
 \alpha_c^{D,\Omega_R}(\rho)=
 -\frac{R}{4\pi(\rho^2-R^2)},
 \qquad
 \alpha_c^{N,\Omega_R}(\rho)=
 \frac1{4\pi}\sum_{\ell=1}^{\infty}
 \frac{\ell}{\ell+1}\frac{R^{2\ell+1}}{\rho^{2\ell+2}}>0.
\end{equation*}
The exterior sphere therefore confirms that Neumann boundary conditions may shift the critical coupling to positive values. At threshold the corresponding zero-energy state has a monopole tail of order $|x|^{-1}$ and is not in $L^2(\Omega_R)$. Here and below, ``threshold resonance'' means such a non-$L^2$ distributional zero-energy solution.

For the exterior disk $\Omega_R=\{x\in\R^2:|x|>R\}$,
$|y|=\rho>R$, the Dirichlet critical value is
\begin{equation*}
 \alpha_c^{D,\Omega_R}(\rho)=
 \frac1{2\pi}\log\frac{\rho^2-R^2}{R}.
\end{equation*}
For a finite constant Robin parameter $\beta>0$, separation into angular
Fourier modes gives the convergent series
\begin{equation}\label{eq:exterior-disk-robin-critical}
 \alpha_c^{\beta,\Omega_R}(\rho)
 =\frac1{2\pi}
 \left[
  \log\frac{\rho^2}{R}+\frac1{\beta R}
  +\sum_{n=1}^\infty
  \frac{n-\beta R}{n(n+\beta R)}
  \left(\frac R\rho\right)^{2n}
 \right].
\end{equation}
Indeed, the addition theorem for $K_0$ diagonalizes the correction problem
in angular Fourier modes. For every $n\ge1$, the contribution of the pair
of modes $\pm n$ to the limit of the regular part is
\begin{equation*}
 \frac1{2\pi n}\frac{n-\beta R}{n+\beta R}
 \left(\frac R\rho\right)^{2n},
\end{equation*}
whereas the zeroth mode, combined with the divergent free diagonal term,
has the finite limit
\begin{equation*}
 \frac1{2\pi}\left(\log\frac{\rho^2}{R}+\frac1{\beta R}\right).
\end{equation*}
Their sum is \eqref{eq:exterior-disk-robin-critical}.
As $\beta\to+\infty$, the series tends to
$\log(1-R^2/\rho^2)$ and \eqref{eq:exterior-disk-robin-critical} reduces
to the Dirichlet value. As $\beta\downarrow0$, the term
$(\beta R)^{-1}$ makes the critical value diverge to $+\infty$. For
Neumann boundary conditions the Weyl function itself diverges to
$+\infty$ as $\kappa\downarrow0$, so there is no finite critical coupling.
Thus the Neumann case is the transition between finite-coupling
threshold crossing and threshold approach at infinite coupling. If negative
Robin parameters are also considered, the interpretation becomes even
sharper. For $\beta<0$, the radial function $K_0(\kappa r)$ is an $L^2$
eigenfunction precisely when
\begin{equation*}
 \kappa K_1(\kappa R)+\beta K_0(\kappa R)=0.
\end{equation*}
This equation has a positive solution for every $\beta<0$, since the
continuous function
$\kappa\mapsto\kappa K_1(\kappa R)/K_0(\kappa R)$ tends to $0$ as
$\kappa\downarrow0$ and to $+\infty$ as $\kappa\to+\infty$. Since the
essential spectrum remains $[0,+\infty)$, the
exterior-disk Robin background has an isolated negative ground state.
Corollary~\ref{cor:isolated-background-ground-state} therefore applies, and
the point-interaction eigenvalue converges to that background eigenvalue.
Hence $\beta=0$ separates the isolated-ground-state regime from the
finite-critical-coupling regime.
\section*{Acknowledgements}
The first author gratefully acknowledges the support of INdAM--GNFM and of the Next Generation EU--PRIN 2022 project ``Singular Interactions and Effective Models in Mathematical Physics'' (2022CHELC7).

\end{document}